\documentclass[envcountsame]{llncs}
\usepackage{hyperref}
\usepackage{amsmath, amssymb}
\usepackage{epsfig,psfrag}
\usepackage{enumerate}

\psfrag{x}{{\footnotesize Payoff of buyer $1$}}
\psfrag{y}{{\footnotesize Payoff of buyer $2$}}
\psfrag{a}{{\footnotesize$(7,8.25)$}}
\psfrag{b}{{\footnotesize$(8,7)$ {\em Fisher Payoff}}}
\psfrag{c}{{\footnotesize$(9.14,3)$}}
\psfrag{l1}{{\footnotesize$L_1$}}
\psfrag{l2}{{\footnotesize$L_2$}}

\long\def\symbolfootnote[#1]#2{\begingroup%
\def\thefootnote{\fnsymbol{footnote}}\footnote[#1]{#2}\endgroup}

\newcommand{\bs}{\boldsymbol}
\DeclareMathOperator*{\argmax}{arg\,max}

\title{Nash Equilibria in Fisher Market}
\author{Bharat Adsul\inst{1} \and Ch. Sobhan Babu\inst{2} \and Jugal Garg\inst{1} \and Ruta Mehta\inst{1} \and\\ Milind Sohoni\inst{1}}
\institute{Indian Institute of Technology, Bombay\\ \email{adsul,jugal,ruta,sohoni@cse.iitb.ac.in}\and Indian Institute of
Technology, Hyderabad\\ \email{sobhan@iith.ac.in}}

\begin{document}
\maketitle

\begin{abstract}
Much work has been done on the computation of market equilibria. However due to strategic play by buyers, it is not clear
whether these are actually observed in the market. Motivated by the observation that a buyer may derive a better payoff by
feigning a different utility function and thereby manipulating the Fisher market equilibrium, we formulate the {\em Fisher
market game} in which buyers strategize by posing different utility functions.  We show that existence of a {\em
conflict-free allocation} is a necessary condition for the Nash equilibria (NE) and also sufficient for the symmetric NE in
this game.
There are many NE with very different payoffs, and the Fisher equilibrium payoff is captured at a symmetric NE.  We provide a
complete polyhedral characterization of all the NE for the two-buyer market game. Surprisingly, all the NE of this game turn
out to be symmetric and the corresponding payoffs constitute a piecewise linear concave curve.  We also study the correlated
equilibria of this game and show that third-party mediation does not help to achieve a better payoff than NE payoffs.  
\end{abstract}

\section{Introduction\label{intro}}
A fundamental market model was proposed by Walras in 1874 \cite{walras}. Independently, Fisher proposed a special
case of this model in 1891 \cite{fisher}, where a market comprises of a set of buyers and divisible goods. The money
possessed by buyers and the amount of each good is specified. The utility function of every buyer is also given.  The market
equilibrium problem is to compute prices and allocation such that every buyer gets the highest utility bundle subject to her
budget constraint and that the market clears. Recently, much work has been done on the computation of market equilibrium
prices and allocation for various utility functions, for example \cite{code,dev,jain,orlin}.

The payoff ({\it i.e.}, happiness) of a buyer depends on the equilibrium allocation and in turn on the utility functions and
initial endowments of the buyers. A natural question to ask is, can a buyer achieve a better payoff by feigning a different
utility function? It turns out that a buyer may indeed gain by feigning! This observation motivates us to analyze the
strategic behavior of buyers in the Fisher market. 
We analyze here the linear utility case described below.

Let $\mathcal{B}$ be the set of buyers, and $\mathcal{G}$ be the set of goods, and $|\mathcal{B}|=m,|\mathcal{G}|=n$. Let
$m_i$ be the money possessed by buyer $i$, and $q_j$ be the total quantity of good $j$ in the market. The utility function of
buyer $i$ is represented by the non-negative utility tuple $\langle u_{i1}, \dots,u_{in}\rangle$, where $u_{ij}$ is
the payoff, she derives from a unit amount of good $j$.  Thus, if $x_{ij}$ is the amount of good $j$ allocated to buyer $i$,
then the {\em payoff} she derives from her allocation is $\sum_{j\in\mathcal G}u_{ij}x_{ij}$. Market equilibrium or market
clearing prices $\langle p_1,\dots,p_n\rangle$, where $p_j$ is the price of good $j$, and equilibrium allocation
$[x_{ij}]_{i\in\mathcal B,j\in\mathcal G}$ satisfy the following constraints:

\begin{itemize}
\item {\bf Market Clearing:} The demand equals the supply of each good, {\it i.e.}, $\forall j\in\mathcal G,\
\sum_{i\in\mathcal B}x_{ij}=q_j$, and $\forall i \in \mathcal{B}$, $\sum_{j\in\mathcal G}p_{j}x_{ij}= m_i$. 
\item {\bf Optimal Goods:} Every buyer buys only those goods, which give her the maximum utility per unit of money, {\it
i.e.}, if $x_{ij}>0$ then $\frac{u_{ij}}{p_j}=\ \max_{k\in\mathcal G}\frac{u_{ik}}{p_k}$.
\end{itemize}

In this market model, 
by scaling $u_{ij}$'s appropriately, we may assume that the quantity of every good is one unit, {\it i.e.}, $q_j=1,\ \forall
j\in\mathcal G$. Equilibrium prices are unique and the set of equilibrium allocations is a convex set \cite{agt}.
The following example illustrates a small market.

\begin{example}\label{eg1}
Consider a 2 buyers, 2 goods market with $m_1=m_2=10$, $q_1=q_2=1$, $\langle u_{11}, u_{12}\rangle =\langle 10, 3\rangle$ and
$\langle u_{21},u_{22}\rangle=\langle 3,10 \rangle$. The equilibrium prices of this market are $\langle
p_1,p_2\rangle=\langle 10,10\rangle$ and the unique equilibrium allocation is $\langle
x_{11},x_{12},x_{21},x_{22}\rangle=\langle 1,0,0,1\rangle$. The payoff of both the buyers is $10$.
\end{example}

In the above market, does a buyer have a strategy to achieve a better payoff? Yes indeed, buyer $1$ can force price change by
posing a different utility tuple, and in turn gain. Suppose buyer $1$ feigns her utility tuple as $\langle 5,15 \rangle$
instead of $\langle 10,3\rangle$, then coincidentally, the equilibrium prices $\langle p_1,p_2\rangle$ are also $\langle
5,15\rangle$. The unique equilibrium allocation $\langle x_{11}, x_{12},x_{21},x_{22}\rangle$ is $\langle
1,\frac{1}{3},0,\frac{2}{3}\rangle$. Now, the payoff of buyer $1$ is $u_{11}*1+u_{12}*\frac{1}{3}=11$, and that of buyer $2$
is $u_{22}*\frac{2}{3}=\frac{20}{3}$. Note that the payoffs are still calculated w.r.t. the true utility tuples.

This clearly shows that a buyer could gain by feigning a different utility tuple, hence the Fisher market is
susceptible to gaming by strategic buyers. Therefore, the equilibrium prices w.r.t. the true utility tuples may
not be the actual operating point of the market. The natural questions to investigate are: What are the possible operating
points of this market model under strategic behavior? Can they be computed? Is there a preferred one? This motivates us
to study the Nash equilibria of the {\em Fisher market game}, where buyers are the players and strategies are the utility
tuples that they may pose. \\
\\
\noindent {\bf Related work.} Shapley and Shubik \cite{ssg} consider a market game for the exchange economy, where every good has a
trading post, and the strategy of a buyer is to bid (money) at each trading post. For each strategy profile, the prices are
determined naturally so that market clears and goods are allocated accordingly, however agents may not get the optimal
bundles. Many variants \cite{amir,dubey} of this game have been extensively studied. Essentially, the goal is to design a
mechanism to implement Walrasian equilibrium (WE), {\it i.e.}, to capture WE at a NE of the game. The strategy space of this
game is tied to the implementation of the market (in this case, trading posts). Our strategy space is the utility tuple
itself, and is independent of the market implementation. It is not clear that bids of a buyer in the Shapley-Shubik game
correspond to the feigned utility tuples. 

In word auction markets as well, a similar study on strategic behavior of buyers (advertisers) has been done
\cite{deng,edelman,varian}.\\

\noindent{\bf Our contributions.} We formulate the Fisher market game, the strategy sets and the corresponding payoff function in
Section \ref{FMG}. Every (pure) strategy profile defines a Fisher market, and therefore market equilibrium prices and a set
of equilibrium allocations.  The payoff of a buyer may not be same across all equilibrium allocations w.r.t. a strategy
profile, as illustrated by Example \ref{eg2} in Section \ref{FMG}.  Furthermore, there may not exist an equilibrium
allocation, which gives the maximum possible payoffs to all the buyers. This behavior causes a conflict of interest among
buyers. A strategy profile is said to be {\em conflict-free}, if there is an equilibrium allocation which gives the maximum
possible payoffs to all the buyers. 

A strategy profile is called a Nash equilibrium strategy profile {\em (NESP)}, if no buyer can unilaterally deviate and get
a better payoff. In Section \ref{gen}, we show that all NESPs are conflict-free. Using the equilibrium prices, we associate a
bipartite graph to a strategy profile and show that this graph must satisfy certain conditions when the corresponding
strategy profile is a NESP. 

Next, we define {\em symmetric} strategy profiles, where all buyers play the same strategy. We show that a symmetric strategy
profile is a NESP iff it is conflict-free. It is interesting to note that a symmetric NESP can be constructed for a given
market game, whose payoff is the same as the Fisher payoff, {\it i.e.}, payoff when all buyers play truthfully. 
Example \ref{eg5} shows that all NESPs need not be symmetric and the payoff w.r.t. a NESP need not be Pareto optimal ({\it
i.e.}, efficient). However, the Fisher payoff is always Pareto optimal (see First Theorem of Welfare Economics \cite{micro}).

Characterization of all the NESPs seems difficult; even for markets with only three buyers. We study two-buyer markets in
Section \ref{NE} and the main results are:
\begin{itemize}
\item All NESPs are symmetric and they are a union of at most $2n$ convex sets.
\item The set of NESP payoffs constitute a piecewise linear concave curve and all these payoffs are Pareto optimal.
The strategizing on utilities has the same effect as differing initial endowments (see Second Theorem of
Welfare Economics \cite{micro}).
\item The third-party mediation does not help in this game.
\end{itemize} 
Some interesting observations about two-buyer markets are:
\begin{itemize}
\item The buyer $i$ gets the maximum payoff among all Nash equilibrium payoffs when she imitates the other, {\it i.e.}, when
they play $(\bs{u_{-i}},\bs{u_{-i}})$, where $\bs{u_{-i}}$ is the true utility tuple of the other buyer. 
\item There may exist NESPs, whose social welfare ({\it i.e.}, sum of the payoffs of both the buyers) is larger than that of
the Fisher payoff (Example \ref{eg6}).
\item For a particular payoff tuple, there is a convex set of NESPs and hence convex set of equilibrium prices. This
motivates a seller to offer incentives to the buyers to choose a particular NESP from this convex set, which fetches the
maximum price for her good. Example \ref{eg7} illustrates this behavior. 
\end{itemize} 

Most qualitative features of these markets may carry over to oligopolies, which arise in numerous
scenarios. For example, relationship between a few manufacturers of aircrafts or automobiles and many suppliers. 
Finally, we conclude in Section \ref{con_fut} that it is highly unlikely that buyers will act according to their true utility
tuples in Fisher markets and discuss some directions for further research.

\section{The Fisher Market Game}\label{FMG}
As defined in the previous section, a linear Fisher market is defined by the tuple $(\mathcal B, \mathcal G,
(\bs{u_i})_{i\in\mathcal B}$, $\bs{m})$, where $\mathcal B$ is a set of buyers, $\mathcal G$ is a set of goods,
$\bs{u_i}=(u_{ij})_{j\in\mathcal G}$ is the true utility tuple of buyer $i$, and $\bs{m}=(m_i)_{i\in\mathcal B}$ is the
endowment vector. We assume that $|\mathcal B|=m, |\mathcal G|=n$ and the quantity of every good is one unit.

The {\em Fisher market game} is a one-shot non-cooperative game, where the buyers are the players, and the {\it strategy set}
is all possible utility tuples that they may pose, {\it i.e.}, $\mathbb S_i= \{\langle s_{i1},s_{i2},\dots,s_{in}\rangle ~ | ~
s_{ij}\ge 0,\ \sum_{j\in\mathcal G}s_{ij}\ne 0 \},\ \forall i\in\mathcal B$. 
Clearly, the set of all strategy profiles is $\mathbb S=\mathbb S_1\times\dots\times \mathbb S_m$. When a strategy profile
$S=(\bs{s_1},\dots,\bs{s_m})$ is played, where $\bs{s_i}\in\mathbb S_i$, we treat $\bs{s_1},\dots,\bs{s_m}$ as utility tuples
of buyers $1,\dots,m$ respectively, and compute the equilibrium prices and a set of equilibrium allocations w.r.t. $S$ and
$\bs{m}$. 

Further, using the equilibrium prices $(p_1,\dots,p_n)$, we generate the corresponding {\em solution graph} $G$ as follows: Let
$V(G)=\mathcal B\cup\mathcal G$. Let $b_i$ be the node corresponding to the buyer $i,\ \forall i\in\mathcal B$ and $g_j$ be
the node corresponding to the good $j,\ \forall j\in\mathcal G$ in $G$. We place an edge between $b_i$ and $g_j$ iff
$\frac{s_{ij}}{p_j}=\max_{k\in\mathcal G}\frac{s_{ik}}{p_k}$, and call the edges of the solution graph as {\em tight edges}.
Note that when the solution graph is a forest, there is exactly one equilibrium allocation, however this is not so, when it
contains cycles. In the standard Fisher market ({\it i.e.}, strategy of every buyer is her true utility tuple), all
equilibrium allocations give the same payoff to a buyer. However, this is not so when buyers strategize on their utility
tuples: Different equilibrium allocations may not give the same payoff to a buyer. The following example illustrates this
scenario.

\begin{example}\label{eg2}
Consider the Fisher market of Example \ref{eg1}. Consider the strategy profile $S=(\langle 1,19\rangle ,\langle
1,19\rangle)$. Then, the equilibrium prices $\langle p_1,p_2\rangle$ are $\langle 1,19\rangle$ and the solution graph is a
cycle. There are many equilibrium allocations and the allocations $[x_{11},x_{12},x_{13},x_{14}]$ achieving the highest
payoff for buyers $1$ and $2$ are $[1,\frac{9}{19},0,\frac{10}{19}]$ and $[0,\frac{10}{19},1,\frac{9}{19}]$ respectively. The
payoffs corresponding to these allocations are $(11.42,5.26)$ and $(1.58,7.74)$ respectively. Note that there is no
allocation, which gives the maximum possible payoff to both the buyers. 
\end{example}

Let $\bs{p}(S)=(p_1,\dots,p_n)$ be the equilibrium prices, $G(S)$ be the solution graph, and $\mathbb X(S)$ be the set of
equilibrium allocations w.r.t. a strategy profile $S$. The payoff w.r.t. $X\in\mathbb X(S)$ is defined as
$(u_1(X),\dots,u_m(X))$, where $u_i(X)=\sum_{j\in{\mathcal G}}u_{ij}x_{ij}$. Let $w_i(S)=\max_{X\in\mathbb
X(S)}u_i(X),\forall i\in\mathcal B$. 

\begin{definition}
A strategy profile $S$ is said to be {\bf \em conflict-free} if $\exists X\in\mathbb X(S)$, s.t. $u_i(X)=w_i(S),\ \forall
i\in\mathcal B$. Such an $X$ is called a {\bf \em conflict-free allocation}.
\end{definition} 

When a strategy profile $S=(\bs{s_1},\dots,\bs{s_m})$ is not
conflict-free, there is a conflict of interest in selecting a particular allocation for the play. If a buyer, say $k$, does
not get the same payoff from all the equilibrium allocations, {\it i.e.}, $\exists X\in\mathbb X(S),\ u_k(X)<w_k(S)$, then we
show that for every $\delta>0$, there exists a strategy profile $S'=(\bs{s'_1},\dots,\bs{s'_m})$, where
$\bs{s'_i}=\bs{s_i},\ \forall i\ne k$, such that $u_k(X')> w_k(S)-\delta, \forall X'\in\mathbb X(S')$ (Section \ref{crp}). The
following example illustrates the same.

\begin{example}\label{eg3}
In Example \ref{eg2}, for $\delta=0.1$, consider $S'=(\langle 1.1,18.9 \rangle$, $\langle1,19\rangle)$,
{\it i.e.}, buyer $1$ deviates slightly from $S$. Then, $\bs{p}(S')=\langle 1.1,18.9\rangle$, and $G(S')$ is a tree; the
cycle of Example \ref{eg2} is broken. Hence there is a unique equilibrium allocation, and $w_1(S')=11.41,\ w_2(S')=5.29$. 
\end{example}

Therefore, if a strategy profile $S$ is not conflict-free, then for every choice of allocation $X\in\mathbb X(S)$ to decide
the payoff, there is a buyer who may deviate and assure herself a better payoff. In other words, when $S$ is not
conflict-free, there is no way to choose an allocation $X$ from $\mathbb X(S)$ acceptable to all the buyers. This suggests
that only conflict-free strategies are interesting. Therefore, we may define the payoff function $\mathcal P_i:\mathbb
S\rightarrow \mathbb R$ for each player $i \in \mathcal B$ as follows:

\begin{equation}
\forall S\in\mathbb S,\ \mathcal P_i(S)=u_i(X), \mbox{ where } X=\displaystyle\argmax_{X'\in\mathbb X(S)}\prod_{i\in\mathcal B}u_i(X').
\end{equation}
Note that the payoff functions are well-defined 
and when $S$ is conflict-free, $\mathcal P_i(S)=w_i(S),\  \forall i \in \mathcal B$. 

\section{Nash Equilibria: A Characterization}\label{gen}
In this section, we prove some necessary conditions for a strategy profile to be a NESP of the Fisher market game defined in
the previous section. Nash equilibrium \cite{Nash} is a solution concept for games with two or more rational players. When a
strategy profile is a NESP, no player benefits by changing her strategy unilaterally. 

For technical convenience, we assume that $u_{ij}>0$ and $s_{ij}>0,\ \forall i\in\mathcal B, \forall j\in\mathcal G$.
The boundary cases may be easily handled separately. Note that if $S=(\bs{s_1},\dots,\bs{s_m})$ is a NESP then
$S'=(\alpha_1\bs{s_1}, \dots, \alpha_m\bs{s_m}$), where $\alpha_1, \dots, \alpha_m>0$, is also a NESP. Therefore, w.l.o.g.
we consider only the normalized strategies $\bs{s_i}=\langle s_{i1},\dots,s_{in}\rangle$, where
$\sum_{j\in\mathcal G} s_{ij}=1\symbolfootnote[1]{For simplicity, we do use non-normalized strategy profiles in the
examples.},\ \forall i\in\mathcal B$. As mentioned in the previous section, the true utility tuple of buyer $i$ is $\langle
u_{i1}, \dots,u_{in}\rangle$. For convenience, we may assume that $\sum_{j\in\mathcal G} u_{ij}=1$ and $\sum_{i\in\mathcal
B}m_i=1$ (w.l.o.g.). 

We show that all NESPs are conflict-free. However, all conflict-free strategies are not NESPs. A symmetric strategy
profile, where all players play the same strategy ({\it i.e.}, $\forall i,j \in \mathcal B,\ {\bs s_i}={\bs s_j}$),
is a NESP iff it is conflict-free. If a strategy profile $S$ is not conflict-free, then there is a buyer $a$ such
that $\mathcal P_a(S)<w_a(S)$. The ConflictRemoval procedure in the next section describes how she may deviate and assure
herself payoff almost equal to $w_a(S)$.
\vspace{-0.3cm}
\subsection{Conflict Removal Procedure}\label{crp}
\begin{definition}
Let $S$ be a strategy profile, $X\in\mathbb X(S)$ be an allocation, and $P=v_1,v_2,v_3,\dots$ be a path in $G(S)$. $P$ is
called an {\bf \em alternating path} w.r.t. $X$, if the allocation on the edges at odd positions is non-zero, {\it i.e.},
$x_{v_{2i-1}v_{2i}}>0,\forall i\ge 1$.
The edges with non-zero allocation are called {\bf \em non-zero edges}.
\end{definition} 

\begin{table}[!hbt]
\vspace{-0.5cm}
\begin{center}
\begin{tabular}{|l|}\hline
{\bf ConflictRemoval}$(S,b_a,\delta)$\\
\hspace{15pt}{\bf while} $b_a$ belongs to a cycle in $G(S)$ {\bf do}\\
\hspace{30pt}$(p_1,\dots,p_n)\leftarrow \bs{p}(S)$;\\
\hspace{30pt}$J\leftarrow \{j\in\mathcal G\ |\ \mbox{the edge } (b_a,g_j)\mbox{ belongs to a cycle in } G(S)\}$; \\
\hspace{30pt}$g_b\leftarrow \displaystyle\argmax_{j\in J}$ $\frac{u_{aj}}{p_j}$;\\
\hspace{30pt}$X\leftarrow$ an allocation in $\mathbb X(S)$ such that $u_a(X)=w_a(S)$ and $x_{ab}$ is maximum;\\
\hspace{30pt}$S\leftarrow$ Perturbation$(S,X,b_a,g_b,\frac{\delta}{n})$;\\
\hspace{15pt}{\bf endwhile}\\
\hspace{15pt}{\bf return} $S$;\\
\\
{\bf Perturbation}($S, X, b_a, g_b, \gamma$)\\
\hspace{15pt}$S'\leftarrow S$;\\
\hspace{15pt}{\bf if} $(b_a,g_b)$ does not belong to a cycle in $G(S)$ {\bf then}\\
\hspace{30pt}{\bf return} $S'$;\\ 
\hspace{15pt}{\bf endif}\\
\hspace{15pt}$J_1\leftarrow \{v\ | \mbox{ there is an alternating path from } b_a \mbox{ to } v \mbox{ in }
G(S)\setminus(b_a,g_b) \mbox{ w.r.t. } X\}$;\\
\hspace{15pt}$J_2\leftarrow \{v\ | \mbox{ there is an alternating path from } g_b \mbox{ to } v \mbox{ in }
G(S)\setminus(b_a,g_b) \mbox{ w.r.t. } X\}$;\\
\hspace{15pt}$(p_1,\dots,p_n)\leftarrow \bs{p}(S)$; $\ \ l\leftarrow \sum_{g_j\in J_1}p_j$; $\ \ r\leftarrow \sum_{g_j\in
J_2}p_j$; \\ 
\hspace{15pt}W.r.t. $\alpha$, define prices of goods to be\\
\hspace{30pt}$\forall g_j\in J_1: (1-\alpha)p_j$; $\ \ \forall g_j\in J_2: (1+\frac{l \alpha}{r})p_j$; $\ \ \forall
g_j\in\mathcal G\setminus(J_1\cup J_2): p_j$;\\ 
\hspace{15pt}Raise $\alpha$ infinitesimally starting from $0$ such that none of the three events occur:\\
\hspace{30pt}Event 1: a new edge becomes tight;\\
\hspace{30pt}Event 2: a non-zero edge becomes zero;\\
\hspace{30pt}Event 3: payoff of buyer $a$ becomes $u_a(X)-\gamma$;\\
\hspace{15pt}$s'_{ab}\leftarrow s_{ab}\frac{(1+\frac{l\alpha}{r})}{(1-\alpha)};\
\bs{s'_a}\leftarrow\frac{\bs{s'_a}}{\sum_{j\in\mathcal G}s'_{aj}}$;\\
\hspace{15pt}{\bf return} $S'$;\\
\hline
\end{tabular}
\end{center}
\caption{Conflict Removal Procedure}\label{a1}
\vspace{-0.6cm}
\end{table}

The ConflictRemoval procedure in Table \ref{a1} takes a strategy profile $S$, a buyer $a$ and a positive number $\delta$, and
outputs another strategy profile $S'$, where $\bs{s'_i}=\bs{s_i},\ \forall i\neq a$ such that $\forall X'\in\mathbb X(S'),\
u_a(X')>w_a(S)-\delta$. The idea is that if a buyer, say $a$, does not belong to any cycle in the solution graph of a
strategy profile $S$, then $u_a(X)=w_a(S),\ \forall X\in\mathbb X(S)$. The procedure essentially breaks all the cycles
containing $b_a$ in $G(S)$ using the Perturbation procedure iteratively such that the payoff of buyer $a$ does not decrease
by more than $\delta$. 

The Perturbation procedure takes a strategy profile $S$, a buyer $a$, a good $b$, an allocation
$X\in\mathbb X(S)$, where $x_{ab}$ is maximum among all allocations in $\mathbb X(S)$ and a positive number $\gamma$, and
outputs another strategy profile $S'$ such that $\bs{s'_i}=\bs{s_i},\ \forall i\neq a$ and $w_a(S')>u_a(X)-\gamma$. It
essentially breaks all the cycles containing the edge $(b_a,g_b)$ in $G(S)$. 

A detailed explanation of both the procedures is given in Appendix \ref{expl}. In the next theorem, we use the ConflictRemoval
procedure to show that all the NESPs in the Fisher market game are conflict-free. 

\begin{theorem}\label{nece}
If $S$ is a NESP, then 
\begin{itemize}
\item[(i)] $\exists X\in\mathbb X(S)$ such that $u_i(X)=w_i(S),\forall i\in\mathcal B$, {\it i.e.}, $S$ is conflict-free. 
\item[(ii)] the degree of every good in $G(S)$ is at least 2.
\item[(iii)] for every buyer $i\in\mathcal B,\ \exists k_i\in K_i$ s.t. $x_{ik_i}>0$, where $K_i=\{j\in\mathcal G\ |\
\frac{u_{ij}}{p_j}=\max_{k\in\mathcal G}\frac{u_{ik}}{p_k}\}$, $(p_1,\dots,p_n)=\bs{p}(S)$ and $[x_{ij}]$ is a conflict-free
allocation. 
\end{itemize} 
\end{theorem} 

\begin{proof}
Suppose there does not exist an allocation $X\in\mathbb X(S)$ such that $u_i(X)=w_i(S),\ \forall i\in\mathcal B$, then there
is a buyer $k\in\mathcal B$, such that $\mathcal P_k(S)<w_k(S)$. Clearly, buyer $k$ has a deviating strategy (apply
ConflictRemoval on the input tuple $(S,k,\delta)$, where $0<\delta<(w_k(S)-\mathcal P_k(S))$), which is a contradiction.

For part (ii), if a good $b$ is connected to exactly one buyer, say $a$, in $G(S)$, then buyer $a$ may gain by reducing
$s_{ab}$, so that price of good $b$ decreases and prices of all other goods increase by the same factor.

For part (iii), if there exists a buyer $i$ such that $x_{ik_i}=0,\ \forall k_i\in K_i$, then she may gain by increasing the
utility for a good in $K_i$. \qed 
\end{proof} 

The following example shows that the above conditions are not sufficient. 

\begin{example}\label{eg4}
Consider a market with 3 buyers and 2 goods, where $\bs{m}=\langle 50,100,$ $50\rangle$, $\bs{u_1}=\langle 2,0.1\rangle,
\bs{u_2}=\langle 4,9\rangle$, and $\bs{u_3}=\langle 0.1,2\rangle$. Consider the strategy profile
$S=(\bs{u_1},\bs{u_2},\bs{u_3})$ given by the true utility tuples. The payoff tuple w.r.t. $S$ is $(1.63,6.5,0.72)$. It
satisfies all the necessary conditions in the above theorem, however $S$ is not a NESP because buyer $2$ has a deviating
strategy $\bs{s'_2}=\langle 2,3\rangle$ and the payoff w.r.t. strategy profile $(\bs{s_1},\bs{s_2'},\bs{s_3})$ is
$(1.25,6.75,0.83)$.
\end{example}

\subsection{Symmetric and Asymmetric NESPs}\label{symmAsymm}
Recall that a strategy profile $S=(\bs{s_1},\dots,\bs{s_m})$ is said to be a {\em symmetric} strategy profile if
$\bs{s_1}=\dots= \bs{s_m}$, {\it i.e.}, all buyers play the same strategy. 
\begin{proposition}\label{symm}
A symmetric strategy profile S is a NESP iff it is conflict-free.
\end{proposition}
\begin{proof}
$(\Rightarrow)$ is easy (Theorem \ref{nece}). For $(\Leftarrow)$, suppose a buyer $i$ may deviate and gain, then the prices
have to be changed. In that case, all buyers except buyer $i$ will be connected to only those goods, whose prices are
decreased. This leads to a contradiction (details are in Appendix \ref{Acorr}).  \qed
\end{proof} 

Let $S^f=[s_{ij}]$ be a strategy profile, where $s_{ij}=u_{ij},\forall i\in\mathcal B, \forall j\in\mathcal G$, {\it i.e.},
true utility functions. All allocations in $\mathbb{X}(S^f)$ give the same payoff to the buyers 
({\it i.e.}, $\forall i \in \mathcal B, u_i(X)=w_i(S^f),\ \forall X\in \mathbb{X}(S^f)$),
and we define {\em Fisher payoff} $(u^f_1,\dots,u^f_m)$ to be
the payoff derived when all buyers play truthfully. 

\begin{corollary}\label{fishpay}
A symmetric NESP can be constructed, whose payoff is the same as the Fisher payoff.
\end{corollary} 
\begin{proof}
Let $S=(\bs{s},\dots,\bs{s})$ be a strategy profile, where $\bs{s}={\bs p}(S^f)$. Clearly $S$ is a symmetric NESP, whose
payoff is the same as the Fisher payoff.  \qed
\end{proof} 

\begin{remark}
The payoff w.r.t. a symmetric NESP is always Pareto optimal. For a Fisher market game, there is exactly one symmetric NESP
iff the degree of every good in $G(S^f)$ is at least two \cite{wm}.
\end{remark}



The characterization of all the NESPs for the general market game seems hard; even for markets with only three buyers. The
following example illustrates an asymmetric NESP, whose payoff is not Pareto optimal. 

\begin{example}\label{eg5}
Consider a market with 3 buyers and 2 goods, where $\bs{m}=\langle 50, 100,$ $50\rangle$, $\bs{u_1}=\langle 2,3\rangle,
\bs{u_2}=\langle 4,9\rangle$, and $\bs{u_3}=\langle 2,3\rangle$. Consider the two strategy profiles given by
$S_1=(\bs{s_1},\bs{s_2},\bs{s_3})$ and $S_2=(\bs{s},\bs{s},\bs{s})$, where $\bs{s_1}=\langle 2,0.1\rangle,\bs{s_2}=\langle
2,3\rangle, \bs{s_3}=\langle 0.1,3\rangle$, and $\bs{s}=\langle 2,3\rangle$. The payoff tuples w.r.t. $S_1$ and $S_2$ are
$(1.25,6.75,1.25)$ and $(1.25,7.5,1.25)$ respectively. Note that both $S_1$ and $S_2$ are NESPs for the above market (details
are in Appendix \ref{Acorr}). 
\end{example}

\section{The Two-Buyer Markets}\label{NE}
A two-buyer market consists of two buyers and a number of goods. These markets arise in numerous scenarios. 
The two firms in a duopoly may be considered as the two buyers with a similar requirements to fulfill from a large number of
suppliers, for example, relationship between two big automotive companies with their suppliers. 

In this section, we study two-buyer market game and provide a complete polyhedral characterization of NESPs, all of which
turn out to be symmetric. Next, we study how the payoffs of the two buyers change with varying NESPs
and show that these payoffs constitute a piecewise linear concave curve. For a particular payoff tuple on this curve, there
is a convex set of NESPs, hence a convex set of equilibrium prices, which leads to a different class of non-market behavior
such as incentives. Finally, we study the correlated equilibria of this game and show that third-party mediation does not
help to achieve better payoffs than any of the NESPs.  

\begin{lemma}\label{char}
All NESPs for a two-buyer market game are symmetric. 
\end{lemma}

\begin{proof}
If a NESP $S=(\bs{s_1},\bs{s_2})$ is not symmetric, then $G(S)$ is not a complete bipartite graph. Therefore there is a
good, which is exclusively bought by a buyer, which is a contradiction (Theorem \ref{nece}, part $(ii)$).\qed
\end{proof}

\subsection{Polyhedral Characterization of NESPs}\label{comp_NE}
In this section, we compute all the NESPs of a Fisher market game with two buyers. Henceforth we assume that the goods are so
ordered that $\frac{u_{1j}}{u_{2j}}\ge \frac{u_{1(j+1)}}{u_{2(j+1)}}$, for $j=1,\dots,n-1$. Chakrabarty et al. \cite{two}
also use such an ordering to design an algorithm for the linear Fisher market with two agents.  Let $S=(\bs{s},\bs{s})$ be a
NESP, where $\bs{s}=(s_1,\dots,s_n)$ and $(p_1,\dots,p_n)=\bs{p}(S)$. The graph $G(S)$ is a complete bipartite graph. Since
$m_1+m_2=1$ and $\sum_{j=1}^ns_{j}=1$, we have $p_j=s_j,\forall j\in\mathcal G$. In a conflict-free allocation $X\in \mathbb
X(S)$, if $x_{1i}>0$ and $x_{2j}>0$, then clearly $\frac{u_{1i}}{p_i}\ge \frac{u_{1j}}{p_j}$ and $\frac{u_{2i}}{p_i}\le
\frac{u_{2j}}{p_j}$.


\begin{definition}
An allocation $X=[x_{ij}]$ is said to be a {\bf \em nice allocation}, if it satisfies the property: $x_{1i}>0$ and
$x_{2j}>0\ \Rightarrow\ i\le j$.
\end{definition}

The main property of a {\em nice allocation} is that if we consider the goods in order, then from left to right, goods
get allocated first to buyer $1$ and then to buyer $2$ exclusively, however they may share at most one good in between. Note
that a symmetric strategy profile has a unique nice allocation. 

\begin{lemma}\label{order}
Every NESP has a unique conflict-free nice allocation.
\end{lemma}
\begin{proof}
The idea is to convert a conflict-free allocation into a nice allocation through an exchange s.t. payoff remains same
(details are in Appendix \ref{Acorr}).\qed
\end{proof}

The non-zero edges in a nice allocation either form a tree or a forest containing two trees. We use the properties of nice
allocations and NESPs to give the polyhedral characterization of all the NESPs. The convex sets $B_k$ for all $1\le
k\le n$, as given in Table \ref{NE1}, correspond to all possible conflict-free nice allocations, where non-zero edges form a
tree, and the convex sets $B_k'$ for all $1\le k\le n-1$, as given in Table \ref{NE2}, correspond to all possible
conflict-free nice allocations, where non-zero edges form a forest\symbolfootnote[1]{In both the tables $\alpha_i$'s may 
be treated as price variables.}. Let $\mathbb B=\displaystyle\cup_{k=1}^n B_k\cup_{k=1}^{n-1}B_k'$ and
$S^{NE}=\{(\bs{\alpha},\bs{\alpha})\ |\ \bs{\alpha}=(\alpha_1,\dots,\alpha_n)\in \mathbb B \}$. Note that $S^{NE}$ is a
connected set. 

\begin{table}[here]
\vspace{-0.2cm}
\begin{minipage}{0.45\textwidth}
\centering
\begin{tabular}{|rcll|}\hline
$\sum_{i=1}^{k-1} \alpha_i$ & $<$ & $m_1$ & \\
$\sum_{i=k+1}^n \alpha_i$ & $<$  & $m_2$ & \\
$\sum_{i=1}^n \alpha_i$ & $=$ & $m_1+m_2$ & \\
$u_{1j}\alpha_i-u_{1i}\alpha_j$ &  $\le$ & $0$ & $\forall i\le k,\forall j\ge k$ \\
$u_{2i}\alpha_j - u_{2j}\alpha_i$ & $\le$ & $0$ & $\forall i\le k,\forall j\ge k$ \\
$\alpha_i$ & $\ge$ & $0$ & $\forall i\in{\mathcal G}$ \\
\hline
\end{tabular}
\caption{$B_k$}\label{NE1}
\end{minipage}
\begin{minipage}{0.6\textwidth}
\centering
\begin{tabular}{|rcll|}\hline
$\sum_{i=1}^k \alpha_i$ & $=$ & $m_1$ &\\
$\sum_{i=k+1}^n \alpha_i$ &  $=$ & $m_2$ &\\
$u_{1j}\alpha_i - u_{1i}\alpha_j$ & $\le$ & $0$ & $\forall i\le k, \forall j\ge k+1$ \\
$u_{2i}\alpha_j - u_{2j}\alpha_i$ & $\le$ & $0$ & $\forall i\le k, \forall j\ge k+1$ \\
$\alpha_i$ & $\ge$ & $0$ & $\forall i\in \mathcal G$ \\
\hline
\end{tabular}
\caption{$B'_k$}\label{NE2}
\end{minipage}
\vspace{-0.7cm}
\end{table}

\begin{lemma}\label{2buyerNESP}
A strategy profile $S$ is a NESP iff $S\in S^{NE}$. 
\end{lemma}
\begin{proof}
($\Leftarrow$) is easy by the construction and Proposition \ref{symm}. For the other direction, we know that every NESP has a
conflict-free nice allocation (Lemma \ref{order}), and $\mathbb B$ corresponds to all possible conflict-free nice
allocations.\qed
\end{proof} 

\subsection{The Payoff Curve}\label{payoff}
In this section, we consider the payoffs obtained by both the players at various NESPs. Recall that whenever a strategy
profile $S$ is a NESP, $\mathcal P_i(S)=w_i(S),\ \forall i\in\mathcal B$. Henceforth, we use $w_i(S)$ as the payoff of buyer
$i$ for the NESP $S$. 
Let $\mathbb F=\{(w_1(S), w_2(S))\ |\ S \in S^{NE}\}$ be the set of all possible NESP payoff tuples.

Let $\mathcal X$ be the set of all nice allocations, and $\mathbb H=\{(u_1(X),u_2(X))\ |\ X\in\mathcal X\}$. 
For $\alpha\in[0,1]$, let  $t(\alpha)=(\langle s_1,\dots,s_n\rangle$, $\langle s_1,\dots,s_n\rangle)$, where
$s_i=u_{1i}+\alpha(u_{2i}-u_{1i})$, and $\mathbb G=\{(w_1(S),w_2 (S))\ |\ S=t(\alpha), \alpha \in [0,1]\}$.

\begin{proposition}\label{prop_curve}
$\mathbb F$ is a piecewise linear concave (PLC) curve.
\end{proposition}
\begin{proof} 
\vspace{-0.13cm}
The proof is based on the following steps (details are in Appendix \ref{Acorr}).
\begin{enumerate}
\item $\mathbb H$ is a PLC curve with $(0,1)$ and $(1,0)$ as the end points.
\item $\forall\alpha\in[0,1],\  t(\alpha)\in S^{NE}$, then clearly $\mathbb G\subset{\mathbb H}$. Since the nice allocation w.r.t.
$t(\alpha)$ changes continuously as $\alpha$ moves from $0$ to $1$, so we may conclude that $\mathbb G$ is a PLC curve with the end
points $(w_1(S^1),w_2(S^1))$ and $(w_1(S^2),w_2(S^2))$, where $S^1=t(0)$ and $S^2=t(1)$.  
\item $\mathbb F=\mathbb G$. \qed
\end{enumerate}
\end{proof}

The next example demonstrates the payoff curve for a small market game.
\vspace{-0.1cm}
\begin{example}\label{eg6}
Consider a market with 3 goods and 2 buyers, where $\bs{m}=\langle 7,3\rangle$, $\bs{u_1}=\langle 6, 2, 2 \rangle$, and
$\bs{u_2}=\langle 0.5, 2.5, 7 \rangle$. The payoff curve for this game is shown in the following figure.

\centerline{
\epsfysize=3.5cm
\epsfxsize=4.8cm
\epsfbox{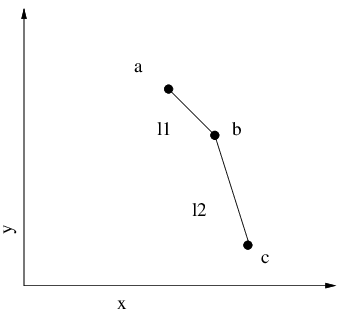}
}

The first and the second line segment of the curve correspond to the sharing of good $2$ and $3$ respectively. 
The payoffs corresponding to the boundary NESPs $S^1=t(0)$ and $S^2=t(1)$ are $(7,8.25)$ and
$(9.14,3)$ respectively. Furthermore, the {\em Fisher payoff} $(8,7)$ may be achieved by a NESP $t(0.2)$. 
Note that in this example the social welfare ({\it i.e.}, sum of the payoffs of both the buyers) from the
Fisher payoff ($15$) is lower than that of the NESP $S^1$ ($15.25$).
\end{example}

\vspace{-0.5cm}
\subsection{Incentives}\label{bribe}
For a fixed payoff tuple on the curve $\mathbb F$, there is a convex set of NESPs and hence a convex set of prices,
giving the same payoffs to the buyers, and these can be computed using the convex sets defined in Table \ref{NE1} and
\ref{NE2}. This leads to a different class of behavior, {\it i.e.}, motivation for a seller to offer incentives to the buyers
to choose a particular NESP from this convex set, which fetches the maximum price for her good. The following example
illustrates this possibility.

\begin{example}\label{eg7}
Consider a market with 2 buyers and 4 goods, where $\bs{m}=\langle 10,10\rangle$, $\bs{u_1}=\langle 4,3,2,1\rangle$, and
$\bs{u_2}=\langle 1,2,3,4\rangle$. Consider the two NESPs given by $S_1=(\bs{s_1},\bs{s_1})$ and $S_2=(\bs{s_2},\bs{s_2})$,
where $\bs{s_1}=\langle \frac{20}{3},\frac{20}{3},\frac{10}{3},\frac{10}{3}\rangle$ and
$\bs{s_2}=\langle\frac{20}{3},\frac{20}{3},\frac{9}{3},\frac{11}{3}\rangle$. Both $S_1$ and $S_2$ gives the payoff 
$(5.5,8)$, however the prices are different, {\it i.e.}, $\bs p(S_1)=\langle \frac{20}{3},\frac{20}{3},\frac{10}{3},$
$\frac{10}{3} \rangle$ and $\bs p(S_2)=\langle \frac{20}{3},\frac{20}{3},\frac{9}{3},\frac{11}{3}\rangle$. Clearly in $S_2$,
good $3$ is penalized and good $4$ is rewarded (compared to $S_1$).
\end{example}

\vspace{-0.7cm}
\subsection{Correlated Equilibria}\label{coE}
\vspace{-0.1cm}
We have seen in the previous section that the two-buyer market game has a continuum of Nash Equilibria, with
very different and conflicting payoffs. This makes it difficult to predict how a particular game will actually play out in
practice, and if there is a different solution concept which may yield an outcome liked by both the players. 

We examine the correlated equilibria framework as a possibility. Recall that according to the correlated equilibria, the
mediator decides and declares a probability distribution $\pi$ on all possible pure strategy profiles $(\bs{s_1},\bs{s_2})
\in \mathbb S_1\times \mathbb S_2$ beforehand. During the play, she suggests what strategy to play to each player privately,
and no player benefits by deviating from the advised strategy.

Let $\Psi=\{ (x,y)\ |\ \exists (x_0 ,y_0)\in \mathbb H, \mbox{ such that } x\leq x_0 \: and \: y\leq y_0\}$, where $\mathbb H$ is the
PLC curve discussed in Section \ref{payoff}. Clearly, $\Psi$ is a convex set. The question we ask: Is there a correlated
equilibrium $\pi$ such that the payoff w.r.t. $\pi$ lies above the curve $\mathbb H$? We continue with our assumption that
$\frac{u_{1j}}{u_{2j}}\ge \frac{u_{1(j+1)}}{u_{2(j+1)}},\forall j<n$.

\begin{lemma}
\label{underg}
For any strategy profile $S=(\bs{s_1},\bs{s_2})$, for every allocation $X\in\mathbb X(S)$, there exists a point $(x_1,x_2)$
on $\mathbb H$ such that $x_1\geq{u_1(X)}$ and $x_2\geq{u_2(X)}$. 
\end{lemma} 
\vspace{-0.2in}
\begin{proof}
Any allocation $X$ may be converted to a nice allocation through an exchange such that no buyer worse off (details are in
Appendix \ref{Acorr}).\qed
\end{proof}

\begin{corollary}\label{corr}
The correlated equilibrium cannot give better payoff than any NE payoff to all the buyers. 
\end{corollary} 

\vspace{-0.1in}
\begin{remark}
\cite{wm} extends this result for the general Fisher market game.
\end{remark}

\vspace{-0.5cm}
\section{Conclusion}\label{con_fut}
\vspace{-0.2cm}
The main conclusion of the paper is that Fisher markets in practice will rarely be played with true utility functions. In
fact, the utilities employed will usually be a mixture of a player's own utilities and her conjecture on the other player's
true utilities. Moreover, there seems to be no third-party mediation which will induce players to play according to their
true utilities so that the true Fisher market equilibrium may be observed. Further, any notion of market equilibrium should
examine this aspect of players strategizing on their utilities. This poses two questions: (i) is there a mechanism which will
induce players into revealing their true utilities? and (ii) how does this mechanism reconcile with the "invisible hand" of
the market? The strategic behavior of agents 
and the question whether true preferences may ever be revealed, has been of intense study in economics
\cite{mas,samuel1,micro}. 
The main point of departure for this paper is that buyers strategize directly on utilities rather than market
implementation specifics, like trading posts and bundles.
Hopefully, some of these analysis will lead us to a more effective computational model for markets.


On the technical side, the obvious next question is to completely characterize the NESPs for the general Fisher market game.
We assumed the utility functions
of the buyers to be linear, however Fisher market is gameable for the other class of utility functions as well. It will be
interesting to do a similar analysis for more general utility functions. 

\newpage
\begin{appendix}

\section{Explanation of Procedures}\label{expl}
{\bf ConflictRemoval Procedure}

\vspace{0.3cm}
It takes a strategy profile $S$, a buyer $a$ and a positive number $\delta$, and
outputs another strategy profile $S'$, where $\bs{s'_i}=\bs{s_i},\ \forall i\neq a$ such that $u_a(X')>w_a(S)-\delta,\
\forall X'\in\mathbb X(S')$. The idea is that if a buyer, say $a$, does not belong to any cycle in the solution graph, then
her payoff is same across all equilibrium allocations, {\it i.e.}, $u_a(X)=w_a(S),\forall X\in\mathbb X(S)$. The procedure
essentially breaks all the cycles containing $b_a$ in $G(S)$ iteratively such that the payoff of buyer $a$ does not decrease
by more than $\delta$. 

This is done by first picking a good $b$ from $J$, which gives the maximum payoff per unit of money to buyer $a$ among all
the goods in $J$, where $J$ is the set of goods $j$, such that the edge $(b_a,g_j)$ belongs to a cycle in $G(S)$. Then, it picks
an allocation $X=[x_{ij}]$ such that $u_a(X)=w_a(S)$ and $x_{ab}$ is maximum among all allocations in $\mathbb X(S)$. It is
easy to check that such an $X$ exists. Finally, using Perturbation procedure, it obtains another strategy profile $S'$, where
$\bs{s'_i}=\bs{s_i},\ \forall i\neq a$ and $w_a(S')>w_a(S)-\frac{\delta}{n}$. The edge $(b_a,g_b)$ does not belong to any cycle
in $G(S')$ and $E(G(S'))\subset E(G(S))$. Then, it repeats the above steps for the strategy profile $S'$ until $b_a$ belongs
to a cycle in the solution graph. Clearly, there may be at most $n$ repetition steps and the final strategy profile $S'$ is
such that $u_a(X')>w_a(S)-\delta,\ \forall X'\in\mathbb X(S')$.\\
\\
{\bf Perturbation Procedure}

\vspace{0.3cm}
It takes a strategy profile $S$, a buyer $a$, a good $b$, an allocation $X\in\mathbb X(S)$, where $x_{ab}$ is maximum among
all allocations in $\mathbb X(S)$ and a positive number $\gamma$, and outputs another strategy profile $S'$ such that
$\bs{s'_i}=\bs{s_i},\ \forall i\neq a$ and $w_a(S')>u_a(X)-\gamma$. It essentially breaks all the cycles containing the edge
$(b_a,g_b)$ in $G(S)$. 

If $(b_a,g_b)\not\in G(S)$, then it outputs $S'=S$. Otherwise, let $J_1$ and $J_2$ be the sets of buyers and goods to which
there is an alternating path w.r.t. $X$ starting from $b_a$ and $g_b$ in $G(S)\setminus (b_a,g_b)$ respectively. Note that
$J_1\cap J_2=\phi$, otherwise there is an alternating path $P$ from $b_a$ to $g_b$ in $G(S)\setminus(b_a,g_b)$, and using $P$
with $(b_a,g_b)$, $x_{ab}$ may be increased and another allocation $X'\in\mathbb X(S)$ may be obtained, where
$x'_{ab}>x_{ab}$, which contradicts the maximality of $x_{ab}$. 

If the prices of goods in $J_1$ are decreased and the prices of goods in $J_2$ are increased, then clearly all the cycles in
$G(S)$ containing the edge $(b_a,g_b)$ break. Such a price change may be forced by increasing $s_{ab}$ infinitesimally. The
procedure first finds such a price change and then constructs an appropriate strategy profile.
Let $(p_1,\dots,p_n)=\bs{p}(S)ٍ,\ l=\sum_{g_j\in J_1}p_j$, $r=\sum_{g_j\in J_2}p_j$ and $\alpha$ be a variable.
The prices are changed as $p_j'=(1-\alpha)p_j$ for $g_j\in J_1$, $p_j'=(1+\frac{l\alpha}{r})p_j$ for $g_j\in J_2$ and
$p_j'=p_j$ for the remaining goods. When $\alpha$ is increased continuously starting from $0$, the corresponding changes in
prices may trigger any of the following three events:

{\bf Event 1:} A new edge may become tight from a buyer outside $J_1$ to a good in $J_1$ or from a
buyer in $J_2$ to a good outside $J_2$. 

{\bf Event 2:} To reflect the price change, the money has to be pulled out from the goods in $J_1$ and transferred to the
goods in $J_2$ through the edge $(b_a,g_b)$. This may cause a non-zero edge in $J_1$ or $J_2$ to become zero. 

{\bf Event 3:} Since the price of good $b$ as well as the allocation on the edge $(b_a,g_b)$ is increasing, the payoff of
buyer $a$ may decrease and become equal to $u_a(X)-\gamma$.

The procedure finds $\alpha>0$ such that none of the three events occur, and constructs a new strategy profile $S'$, where
$s'_{ij}=s_{ij}, \forall (i,j)\neq (a,b)$, and $s'_{ab}=s_{ab}\frac{(1+\frac{l\alpha}{r})}{(1-\alpha)}$. Clearly,
$\bs{p}(S')=(p'_1,\dots,p'_n)$, and $G(S')=(G[J_1]\cup G[J_2]\cup G[V(G)\setminus (J_1\cup J_2)])+(b_a,g_b)$, where 
$G[U]$ denotes the induced graph on $U\subseteq \mathcal B\cup \mathcal G$ in $G(S)$.

\section{Proofs}\label{Acorr}
\noindent{\bf Proof of Proposition \ref{symm}}\\

\noindent $(\Rightarrow)$ is straightforward (Theorem \ref{nece}). For the other direction, let $S=(\bs{s},\dots,\bs{s})$ be
a conflict-free symmetric strategy profile. Clearly, $G(S)$ is a complete bipartite graph and $\mathcal P_i(S)=w_i(S),\
\forall i \in \mathcal B$. If $S$ is not a NESP, then there is a buyer, say $k$, who may deviate and get a better payoff. Let
$S'=(\bs{s_1'},\dots,\bs{s_m'})$ be a strategy profile, where $\bs{s_i'}=\bs{s},\ \forall i\neq k$, such that $\mathcal
P_k(S')>\mathcal P_k(S)$. Let $X'\in X(S')$ be such that $u_k(X')=\mathcal P_k(S')$. 

Let $(p_1,\dots,p_n)=\bs{p}(S)$ and $(p_1',\dots,p_n')=\bs{p}(S')$. If $\bs{p}(S')=\bs{p}(S)$, {\it i.e.}, $p_j=p_j',\forall
j\in\mathcal G$, then $u_k(X')\le w_k(S)$. For the other case $\bs{p}(S')\ne \bs{p}(S)$, let $J_1=\{j\in\mathcal G\ |\
p'_j<p_j\}$, $J_2=\{j\in\mathcal G\ |\ p'_j=p_j\}$, and $J_3=\{j\in\mathcal G\ |\ p'_j>p_j\}$. Note that all buyers except
$k$ will have edges only to the goods in $J_1$ in $G(S')$, {\it i.e.}, goods whose prices have been decreased. From $X'$, we can
construct an allocation $X\in\mathbb X(S)$, such that $\forall j\in J_2\cup J_3,\ x_{kj}=x'_{kj}$ and $\forall j\in
J_1,\ x_{kj}>x'_{kj}$. Hence $u_k(X')<u_k(X)\le w_k(S)$, which is a contradiction. \qed

\vspace{0.5cm}
\noindent {\bf Explanation of Example \ref{eg5}}\\
\\
\noindent $S_2$ is clearly a NESP because $S_2$ is a symmetric strategy profile and there is a conflict-free allocation
(Proposition \ref{symm}). 

For $S_1$, it can be easily checked that buyers $1$ and $3$ have no deviating strategy. Buyer $2$ is essentially the price
setter. However, no matter whatever the strategy, buyer $2$ plays, buyer $1$ will buy only good $1$ and buyer $3$ will buy
only good $3$. Hence, let $h_2=\frac{4*x}{50+x}+\frac{9*(100-x)}{150-x}$ be the payoff of buyer $2$, when she gives $x$
amount of money to good $1$.  Note that $h_2$ is a {\em concave} curve. We compute the maximum value of $h_2$ w.r.t.
$x\in[0,100]$, which turns out to be $7.5$ for $x=30$. Hence buyer $2$ also has no deviating strategy at $S_1$.

\vspace{0.3cm}
\noindent{\bf Proof of Lemma \ref{order}}\\
\\
Let $S$ be a NESP, which does not have a conflict-free nice allocation. Consider a conflict-free allocation $X\in\mathbb
X(S)$. Then w.r.t.  $X$, there are goods $i$ and $j$ such that $x_{1i}>0$, $x_{2j}>0$ and $i>j$.

Since we started with a conflict-free allocation, $\frac{u_{1i}}{p_i}\ge \frac{u_{1j}}{p_{j}}$, $\frac{u_{2i}}{p_i}\le
\frac{u_{2j}}{p_{j}} \Rightarrow$ $\frac{u_{1j}}{u_{2j}} \le \frac{u_{1i}}{u_{2i}}$. Moreover, we assumed that
$\frac{u_{1j}}{u_{2j}} \ge \frac{u_{1i}}{u_{2i}}$, hence $\frac{u_{1i}}{p_i}=\frac{u_{1j}}{p_{j}}$ and
$\frac{u_{2j}} {p_j}=\frac{u_{2i}}{p_{i}}$. 
Hence, buyer $1$ may take away some money from good $i$ and spend it on good $j$ and buyer $2$ may take away some money from
good $j$ and spend it on good $i$ without affecting the payoffs. This gives another conflict-free allocation. We may repeat
this operation till we get a {\em nice allocation}. 

There is exactly one nice allocation in $\mathbb X(S)$, hence the uniqueness follows.\qed

\vspace{0.3cm}
\noindent{\bf Proof of Proposition \ref{prop_curve}}\\

\noindent Let $\mathbb F$, $\mathcal X$, $\mathbb H$, $t(\alpha)$, $\mathbb G$, $S^1$ and $S^2$ be as defined in Section \ref{payoff}. Let $T_k$ be the tree,
where buyer $1$ and $2$ are adjacent to goods $1,\dots,k$ and $k,\dots,n$ respectively, and $F_k$ be the forest, where buyer
$1$ and $2$ are adjacent to goods $1,\dots,k$ and ${k+1},\dots,n$ respectively. Let $F_0$ be the forest where buyer $1$ is
not adjacent to any good and buyer $2$ is adjacent to all the goods, and $F_n$ is defined similarly. Let $G(X)=(\mathcal
B,\mathcal G, E)$ be the bipartite graph, where the edge $(i,j)\in E$ iff $x_{ij}>0$. Let $t_i=\{(u_1(X),u_2(X))\ |\
G(X)=T_i,\ X\in{\mathcal X}\}$, and $f_i=\{(u_1(X),u_2(X))\ |\ G(X)=F_i,\ X\in{\mathcal X}\}$. 

\begin{claim}\label{conc}
The set $\mathbb H$ is a piece-wise linear concave curve, whose end points are $(0,1)$ and $(1,0)$.
\end{claim} 
\begin{proof}
The proof is based on the following observations:
\begin{itemize}
\item $\mathbb H=\displaystyle\bigcup_{k=1}^n t_i \cup \bigcup_{k=0}^{n} f_i$.
\item $f_i$ is a point, for all $0\le i\le n$, and $f_0=(0,1), f_n=(1,0)$.
\item $t_i$ is a straight line with slope $-\frac{u_{2i}}{u_{1i}}$, and the limit of the end points of $t_i$ are $f_{i-1}$
and $f_{i}$ for all $1\le i\le n$.
\item Since $\frac{u_{1i}}{u_{2i}}\ge \frac{u_{1(i+1)}}{u_{2(i+1)}}$, for all $i<n$, hence slope of $t_i$ decreases as $i$
goes from $1$ to $n$. \qed
\end{itemize} 
\end{proof} 

\begin{claim}\label{line1}
$t(\alpha)\in S^{NE},\forall\alpha\in[0,1]$.
\end{claim}
\begin{proof}
The equilibrium prices w.r.t. $S=t(\alpha)$ are $s_1,\dots,s_n$. Since $\frac{u_{1i}}{u_{2i}} \ge
\frac{u_{1(i+1)}}{u_{2(i+1)}} \Rightarrow \frac{u_{1i}}{s_i} \ge \frac{u_{1(i+1)}}{s_{i+1}},\forall i< n$. We can also view
$s_i=u_{2i}+(1-\alpha)(u_{1i}-u_{2i})$, $\forall i \le n$, and hence $\frac{u_{2i}}{s_i} \le
\frac{u_{2(i+1)}}{s_{i+1}},\forall i<n$. It implies that there exists a conflict-free nice allocation w.r.t. $S$, hence $S\in
S^{NE}$ (Lemma \ref{2buyerNESP}). \qed
\end{proof}

Clearly, $\mathbb G\subset \mathbb H$ and is PLC curve with the end points $(w_1(S^1),w_2(S^1))$ and $(w_1(S^2),w_2(S^2))$. Let $X$ be an
allocation, and $X_1=[x_{1j}],X_2=[x_{2j}]$ be the restrictions of $X$ to buyers $1$ and $2$ respectively. For a NESP $S$,
let $X(S)$ be the (conflict-free) nice allocation. Let $X_1(S)$ and $X_2(S)$ be the restrictions of $X(S)$ to buyers $1$ and
$2$ respectively. The following lemma proves that $\mathbb F$ equals $\mathbb G$ as sets. As a preparation towards this
lemma, we introduce a notion of {\em buyer $i$ getting more goods according to the allocation $X$ than the allocation $X'$},
by which we mean $X'_i\ge X_i$, {\it i.e.}, $x_{ij} \ge x'_{ij}, \forall{j \in \mathcal G}$, and $X'_i\neq X_i$.

\begin{claim}\label{curve_f}
As sets, $\mathbb F=\mathbb G$, {\it i.e.}, if $S\in S^{NE}$ then $(w_1(S),w_2(S)) \in \mathbb G$.
\end{claim}
\begin{proof}
Buyer $1$ gets more goods according to the {\em nice allocation} w.r.t. $S \in S^{NE}$ than the {\em nice allocation} w.r.t.
$S' \in S^{NE}$, {\it i.e.}, $X_1(S)\ge X_1(S')$ and $X_2(S)\le X_2(S')\Leftrightarrow w_1(S) > w_1(S')$ and $w_2(S) < w_2(S')$.
Hence, to show that w.r.t. $S\in S^{NE}$, $w_1(S^1)\le w_1(S)\le w_1(S^2)$ and $w_2(S^2)\le w_2(S)\le w_2(S^1)$, it is enough
to show that $X_1(S^1)\le X_1(S) \le X_1(S^2)$ and $X_2(S^2)\le X_2(S) \le X_2(S^1)$.

Suppose there exists a NESP $S$ such that $w_1(S)<w_1(S^1)$. Clearly, $\bs{p}(S^1)=(u_{11},\dots,u_{1n})$.
Let $(p_1,\dots,p_n)=\bs{p}(S)$.
Consider the {\it nice allocation} w.r.t. $S$, {\it i.e.}, $X(S)$. Since $w_1(S)<w_1(S^1)\Rightarrow X_1(S)<X_1(S^1)$.
This implies price of at least one good, say $a$, allocated to buyer $1$ w.r.t. $X_1(S)$ is more than $u_{1a}$, {\it i.e.},
$p_a>u_{1a}$. Similarly, price of at least one good, say $b$, allocated to buyer $2$ w.r.t. $X_2(S)$ is less than $u_{1b}$,
{\it i.e.}, $p_b<u_{1b}$.

Since $\frac{u_{1a}}{p_a}<\frac{u_{1b}}{p_b}$, buyer $1$ prefers good $b$, which is allocated to buyer $2$, over good $a$. It
is a contradiction, since it violates the property of nice allocation. 
Furthermore, $X_1(S)\le X_1(S^2)$ may follow from the similar argument as above. It implies that $X(S)= X(t(\alpha))$ for
some $\alpha\in[0,1]$. \qed
\end{proof}

\vspace{0.3cm}
\noindent{\bf Proof of Lemma \ref{underg}}\\
\\
\noindent If $X$ is a nice allocation, then clearly $(u_1(X),u_2(X))\in \mathbb H$. Otherwise, there exist two goods $i$ and $j$ such
that $i<j$, $x_{1i}<1$ and $x_{2j}<1$ and $i$ is the smallest such number and $j$ is the largest such number. The following
equations show that there exist $z$ and $w$, such that buyer $1$ can exchange $w$ amount of good $j$ with $z$ amount of good
$i$ with buyer $2$ and payoffs of both the buyers are not decreased: 

\begin{center}
$\frac{u_{1j}}{u_{2j}}\leq \frac{u_{1i}}{u_{2i}}\ \Rightarrow\ \frac{u_{2i}}{u_{2j}}\le\frac{w}{z}\leq \frac{u_{1i}}{u_{1j}}
\ \Rightarrow\ zu_{1i}-wu_{1j}\geq{0}$ and $-zu_{2i}+wu_{2j}\geq{0}$ 
\end{center} 

Therefore, they exchange goods $i$ and $j$ in a ratio $z$ and $w$ satisfying above equations till either buyer $1$ gets the
entire good $i$ or buyer $2$ gets the entire good $j$. If current allocation is still not nice, then by repeating the above
exchange procedure, we are guaranteed to reach at a nice allocation, whose payoff $(x_1,x_2)\in \mathbb H$ is such that
$x_1\geq{u_1(X)}$ and $x_2\geq{u_2(X)}$. \qed
%
%
\end{appendix}


\vspace{-0.3cm}
\begin{thebibliography}{99}

\bibitem{amir} \textsc{R. Amir, S. Sahi, M. Shubik, and S. Yao}. A strategic market game with complete markets, \emph{Journal
of Economic Theory}, 51:126--143, 1990.


\bibitem{fisher} \textsc{W. C. Brainard, and H. E. Scarf}. How to compute equilibrium prices in 1891, \emph{Cowles Foundation}
Discussion Paper-1272, 2000.

\bibitem{deng} \textsc{T. Bu, X. Deng, Q. Qi}, Forward looking Nash equilibrium for keyword auction, \emph{Inf. Process.
Lett.}, 105(2):41-46, 2008.

\bibitem{two} \textsc{D. Chakrabarty,  N. Devanur, and V.V. Vazirani}. New Results on Rationality and Strongly Polynomial
Solvability in Eisenberg-Gale Markets, \emph{WINE}, 2006.

\bibitem{code} \textsc{B. Codenotti, S. Pemmaraju, and K. Varadarajan}.  On the polynomial time computation of equilibria for
certain exchange economies, \emph{SODA'05}, 2005.


\bibitem{dev} \textsc{N. Devanur, C.H. Papadimitriou, A. Saberi, and V.V. Vazirani}. Market equilibrium via a primal-dual
type algorithm, \emph{Journal of ACM}, 55(5), 2008.


\bibitem{dubey} \textsc{P. Dubey, and J. Geanakoplos}. From Nash to Walras via Shapley-Shubik, \emph{Journal of Mathematical
Economics}, 39:391--400, 2003.


\bibitem{edelman} \textsc{B. Edelman, M. Ostrovsky, and M. Schwarz}. Internet advertising and the generalized second-price
auction: Selling billions of dollars worth of keywords, \emph{The American Economic Review}, (2007), 97(1):242--259.
\bibitem{wm} \textsc{J. Garg}. Nash Equilibria in Fisher Market, \emph{Working Manuscript}, 2010.



\bibitem{jain} \textsc{K. Jain}.  A polynomial time algorithm for computing the Arrow-Debreu market equilibrium for linear
utilities, \emph{FOCS'04}, 2004.

\bibitem{mas} \textsc{A. Mas-Colell, M.D. Whinston, and J.R.Green}. Microeconomic Theory, \emph{Oxford University Press},
1995.

\bibitem{Nash} \textsc{J. F. Nash}. Equilibrium points in n-person games, \emph{Proc. of the National Academy of Sciences of
the United States of America}, 36(1):48--49, 1950.

\bibitem{agt} \textsc{N. Nisan, T. Roughgarden, E. Tardos, and V.V. Vazirani}.  Algorithmic game theory, \emph{Cambridge
University Press}, Cambridge, 2007.

\bibitem{orlin} \textsc{J.B. Orlin}. Improved Algorithms for Computing Fisher's Market Clearing Prices, \emph{STOC'10}, 2010.

\bibitem{samuel} \textsc{P.A. Samuelson} A Note on the Pure Theory of Consumers' Behaviour, {\emph Economica}, 5:61--71, 1938.

\bibitem{samuel1} \textsc{P.A. Samuelson} Foundations of Economic Analysis, {\emph Harward University Press}, 1947.

\bibitem{ssg} \textsc{L. Shapley, and M. Shubik}. Trade using one commodity as a means of payment, \emph{Journal of Political
Economy}, 85(5):937--968, 1977.


\bibitem{varian} \textsc{H. Varian}. Position auctions, \emph{International Journal of Industrial Organization},
25:1163--1178, 2007.

\bibitem{micro} \textsc{H. Varian}. Microeconomic Analysis, Third Edition, 1992.

\bibitem{walras}\textsc{L. Walras}, Elements of Pure Economics, translated by \textsc{Jaff\'{e}, Allen \& Urwin}, {\emph
London}, 1954.

%
\end{thebibliography}
\end{document}